\documentclass[letterpaper,11pt]{article}

\usepackage{amssymb}
\usepackage{graphicx}

\usepackage[latexsym]{}
\usepackage{amssymb}
\usepackage{amsmath}
\usepackage{graphicx}
\usepackage{epsfig}
\usepackage{url}
\usepackage{algorithm}
\usepackage{algorithmic}
\usepackage{stmaryrd}
\usepackage{qtree}

\usepackage{authblk} %authors and affiliation

\usepackage{multicol}

\usepackage{amsthm}

\usepackage{amsthm} %theorem environments
\theoremstyle{plain}% default
\newtheorem{theorem}{Theorem}
\newtheorem{lemma}{Lemma}

\newtheorem{corollary}{Corollary}
\newtheorem{definition}{Definition}

\newcommand\rsb{\rule{0pt}{2.6ex}} %row space before
\newcommand\rsa{\rule[-1.2ex]{0pt}{0pt}} %row space after
\newcommand\rs{\rsb \rsa} %row space before and after

\newcommand{\wg}{\text{SPARQL}^{\mbox{\tiny{W3C}}}}

\newcommand{\core}{\text{SPARQL}^{\mbox{\tiny{R}}}}

\newcommand{\dom}{\operatorname{dom}}
\newcommand{\range}{\operatorname{range}}
\newcommand{\false}{\operatorname{\emph{false}}}
\newcommand{\true}{\operatorname{\emph{true}}}
\newcommand{\error}{\operatorname{\emph{error}}}

\newcommand{\SELECT}{\operatorname{SELECT}}

\newcommand{\WHERE}{\operatorname{WHERE}}
\newcommand{\AAND}{\operatorname{AND}}
\newcommand{\UNION}{\operatorname{UNION}}
\newcommand{\OPT}{\operatorname{OPT}}
\newcommand{\OPTIONAL}{\operatorname{OPT}}
\newcommand{\FILTER}{\operatorname{FILTER}}
\newcommand{\MINUS}{\operatorname{MINUS}}
\newcommand{\DIFF}{\operatorname{DIFF}}

\newcommand{\bound}{\operatorname{bound}}
\newcommand{\EXCEPT}{\operatorname{EXCEPT}}

\def\lojoin{\hbox{\raise -.2em\hbox to-.32em{$\urcorner$} \hbox to-.08em{$\lrcorner$} $\Join$}\,}
\newcommand{\ev}[2]{\llbracket #1\,\rrbracket_{#2}}

\newcommand{\var}{\operatorname{var}}
 
\newcommand{\desc}{\operatorname{desc}}
 
\newcommand{\gp}{\operatorname{gp}}

\newcommand{\cardi}[2]{\operatorname{card}(#1,#2)} %%comando para escribir card[]()

\newcommand{\dlog}{\text{nr-Datalog}^{\neg}}

\newcommand{\facts}{\operatorname{facts}}

\newcommand{\claudio}[1]{}
\newcommand{\cclaudio}[1]{}

\newcommand{\borrado}[1]{}

\usepackage{url}

\begin{document}
\pagestyle{plain}

\title{The multiset semantics of SPARQL patterns\footnote{This is an extended and updated version of the paper accepted at the International Semantic Web Conference 2016.}}
\author[1,3]{Renzo Angles}
\author[2,3]{Claudio Gutierrez}
\affil[1]{Dept. of Computer Science, Universidad de Talca, Chile}
\affil[2]{Dept. of Computer Science, Universidad de Chile, Chile}
\affil[3]{Center for Semantic Web Research (CIWS)}
\date{}

\maketitle

\begin{abstract}
The paper determines the algebraic and logic structure of
  the multiset semantics of the core patterns of SPARQL. 
We prove that the fragment formed by 
AND, UNION, OPTIONAL, FILTER, MINUS and SELECT corresponds
precisely to both, the intuitive multiset relational algebra 
(projection, selection, natural join, arithmetic union and except),
and the multiset  non-recursive Datalog with safe negation.
\end{abstract}

\section{Introduction}

 The incorporation of multisets (also called 
``duplicates'' or ``bags'')\footnote{
  There is no agreement on terminology~(\cite{Melton}, p. 27).
  In this paper we will use the word ``multiset''.
} 
into the semantics of query languages like SQL or SPARQL 
is essentially due to practical concerns: 
 duplicate elimination is expensive and duplicates might be required for
some applications, e.g. for aggregation. Although 
this design decision in SQL may be debatable  (e.g. see~\cite{Date}), 
today multisets are an established fact  in database systems~\cite{Green09,Lamperti}.

  The theory behind these query languages is relational algebra or
equivalently, relational calculus, formalisms that for sets have a clean 
and intuitive semantics  for users, developers and theoreticians~\cite{Abiteboul-Book}.
 The same cannot be said of their extensions to multisets, whose theory 
is complex (particular containment of queries) and their practical
use not always clear for users and developers~\cite{Green09}. 
 Worst, there exist several possible ways of extending set relational 
operators to multisets and one can find them in practice.
 As illustration, let us remind the behaviour of SQL relational
operators. Consider as example the multisets $A= \{ a,a,a, b,b, d,d\}$ 
and $B = \{ a, b, b, c  \}$. Then 
$A$ {\tt UNION ALL} $B = \{ a,a,a,a, b,b,b,b, c, d,d \}$,
that is, the ``sum'' of all the elements in both multisets.
{\tt UNION DISTINCT} is classical set union: $\{ a,b,c,d \}$.
$A$ \text{\tt INTERSECT ALL} $B$ is $\{ a, b, b \}$, i.e., the common
elements in $A$ and $B$, each 
with the minimum of the multiplicities in $A$ and $B$.
Regarding negation or difference, there are at least two: 
$A$ {\tt EXCEPT ALL} $B$ is $\{ a,a, d,d \}$, i.e. the
arithmetical difference of the copies, and
$A$ {\tt EXCEPT} $B$ is $\{ d,d \}$, the elements in $A$ (with their multiplicity) after
filtering out  all elements occurring in $B$. 
The reader can imagine that the ``rules'' for  combining these operators
are not simple nor intuitive as they do not follow the rules of classical set operations. 
 
 Is there a rationale behind the possible extensions? Not easy to tell.
 Early on Dayal et al. \cite{91030} presented two conceptual
approaches to extend the set operators of union, intersection and negation, 
 corresponding to the two possible interpretations of multiple copies of a
  tuple. The first approach treats all copies of a given tuple as
  being identical or indistinguishable. The second one treats all
  copies of a tuple as being distinct, e.g., as having an underlying
  identity. 
  Each of these interpretations gives rise to a different semantics
for multisets.
The first one permits to extend the lattice algebra structure of
sets induced by the $\subseteq$-order by defining a multiset order
$\subseteq_m$ defined as $A \subseteq_m B$ if each element in $A$ is contained in $B$ and
its multiplicity in $B$ is bigger than in $A$. This order gives  
a  lattice meet (multiset intersection) defined 
as the elements $c$ present in both multisets,
and with multiplicity $\min(c_A,c_B)$,  where
$c_A,c_B$ are the number of copies of $c$ in $A$ and $B$ respectively.
This is the {\tt  INTERSECT   ALL} operator  of SQL.  
The lattice join of two multisets gives a union defined as the elements
$c$  present in both multisets with multiplicity $\max(c_A,c_B)$.
This operator is not present in SQL.  As was shown by Albert \cite{90818},
there is no natural negation to add to this lattice to get a 
Boolean algebra structure like in sets.
 The second interpretation (all copies of an element are distinct) gives a
 poor algebraic structure. The union gives in this case an arithmetic
 version, where the elements in the union of the multisets $A$ and $B$ 
are  the elements $c$ present in both multisets with $c_A + c_B$
copies. This is the  {\tt UNION ALL} operator in SQL.  Under this
interpretation, the intersection loses its meaning (always gives the
empty set)  and the difference becomes trivial  ($A - B = A$).

  In order to illustrate the difficulties of having a ``coherent'' group of
operators for multisets, let us summarize the case of SQL, that does not
have a clear rationale on this point.\footnote{
We follow the semantics of ANSI and ISO SQL:1999 Database Language Standard.
}
 We classified the operators under those that:  keep the set semantics;
preserve the lattice structure of multiset order; 
do arithmetic with multiplicities. Let $A, B$ be multisets, and for each element
$c$, let $c_A$ and $c_B$ be their respective multiplicities in $A$ and $B$.
\begin{align*}
\text{union :} &
  \begin{cases}
     set   & \text{\tt UNION DISTINCT }  \text{ (multiplicity: 1)}\\
     lattice & \text{not present in SQL(*) }  \text{ (multiplicity: $\max(c_A,c_B))$}\\
      arithmetic &    \text{{\tt UNION ALL} }    \text{ (multiplicity: $c_A +c_B$)}  
  \end{cases}
\\
\text{intersection :} &
  \begin{cases}
     set   & \text{\tt INTERSECT DISTINCT }    \text{ (multiplicity: 1)} \\
     lattice  &     \text{\tt INTERSECT ALL }    \text{ (multiplicity: $\min(c_A,c_B))$}\\
      arithmetic  &   \text{does not make sense}
  \end{cases}
\\
\text{difference :} &
  \begin{cases}
      set & \text{not present in SQL}(**)
   \text{ (multiplicity: 1)} \\
      lattice & \text{does not exists }   \\
      arithmetic&      \text{\tt EXCEPT ALL}  \text{ (multiplicity:
        $\max(0,c_A - c_B)$}) \\
      filter &  \text{\tt EXCEPT}   \text{ (multiplicity:
        if ($c_B=0$) then $c_A$ else $0$}) \\ 
  \end{cases}
\end{align*}
(*) Can be simulated as {\tt (A UNION ALL B) EXCEPT ALL (A INTERSECT ALL B)}.    \\
(**) Can be simulated as {\tt SELECT DISTINCT * FROM (A EXCEPT  B)}.

\bigskip

 At this point, a  question arises: Are there ``reasonable'', 
``well behaved'',  ``harmonic'', groups of these operations for multisets?
 The answer is positive. Albert \cite{90818} proved that lattice union and
 lattice intersection  plus a filter  difference work well in certain domains. 
On the other hand, Dayal et al.~\cite{91030}  introduced
the multiset versions for projection ($\pi_X$), selection
($\sigma_C$), join ($\bowtie$) and distinct ($\delta$) and
studied their interaction with Boolean operators.
They  showed that the lattice versions above
combine well with selection ($\sigma_{P \vee Q}(r) =  \sigma_{P}(r)
\cup  \sigma_{Q}(r)$ and similarly for intersection);  that 
the arithmetic versions combine well with projection
($\pi_X(r \uplus s) =  \pi_{X}(r) \uplus \pi_{X}(s)$).
   An important facet is the complexity introduced by the different
operators.  Libkin and Wong \cite{90191,90189} and 
Grumbach et al. \cite{90190} studied the expressive power and
complexity of the operations  of the fragment including lattice union
and intersection; arithmetic difference; and distinct.

  For our purposes here, namely the study of the semantics of 
multisets in SPARQL, none of the above fragments help. 
 It turns out that is a formalism
coming from a logical field,  the well behaved fragment
of {\em non-recursive Datalog with safe negation} (nr-Datalog$^\neg$),
the one that matches the semantics of multisets in SPARQL. More precisely,
the natural extension of the usual (set) semantics of Datalog to multisets 
developed by Mumick et al. \cite{90820}.
  In this paper we work out the relational counterpart of this
  fragment, inspired by the framework
defined by Dayal et al.~\cite{91030}, 
and come up with a {\em  Multiset Relational Algebra} (MRA)
that captures precisely the multiset semantics of
the core relational patterns of SPARQL. MRA is based on the
operators projection ($\pi$), selection ($\sigma$), 
natural join ($\bowtie$), union ($\uplus$) and filter difference ($\setminus$).
The identification of this algebra and the
proof of the correspondence with a relational core of SPARQL 
are  the main contributions of this paper.
 Not less important, as a side effect,  
this approach gives a new relational view of
SPARQL (closer to classical relational algebra and hence more
intuitive for people trained in SQL);  allows to make a clean translation to a 
logical framework (Datalog); and matches precisely the fragment of 
SQL  corresponding to it.  Table 1 shows a glimpse of
these correspondences, whose details are worked in this paper.

\begin{table}[!]
\centering
\caption{{\sc Schema of correspondences among:} 
Multiset SPARQL pattern operators; Multiset Relational Algebra operators;
Datalog rules; and SQL expressions. The operator {\tt EXCEPT} in 
SPARQL is new (although expressible), $\uplus$ is arithmetic union,
and $\setminus$ in MRA is the multiset filter difference.\smallskip}
\begin{tabular}{|l|c|l|l|}
\hline
 SPARQL & Multiset &  $\dlog$ & \hspace{.7cm} SQL \\ 
        & Relational    &     & \\ 
& algebra & & \\ \hline \hline        
        
{\tt SELECT X ...}  & $\pi_X(...)$ & $q(X) \gets
                                                       L_1,\dots,L_n$
                                                                    & {\tt SELECT  X  ...}
                                    \\ \hline
{\tt {P FILTER C }}  & $\sigma_{C}(r)$ 
                             & $L \gets L_P, C$ & {\tt FROM r WHERE  C } \\ \hline
{\tt {P1 . P2 }}   & $r_1 \bowtie r_2$ 
& $L \gets L_1, L_2$ &  {\tt r1 NATURAL JOIN r2}   \\ \hline
{\tt P1 UNION P2 }  & $r_1 \uplus r_2$  &
                                                                       $L
                                                                  \gets
                                                                    L_1$
  & {\tt r1 UNION  ALL  r2} \\
 & &  $L \gets L_2$ &
\\ \hline 
{\tt P1 EXCEPT P2}  & $r_1 \setminus
                                                     r_2$ 
& $L \gets L_1, \neg L_2$ & {\tt r1 EXCEPT  r2} \\ \hline
\end{tabular}
\label{table:equivalences}
\end{table}

\paragraph{Contributions.}
 Summarizing,  this paper advances the current
understanding of the SPARQL language by determining the
precise algebraic (Multiset Relational Algebra) and logical (nr-Datalog$^{\neg}$) 
structure of the multiset semantics
of the core pattern operators in the language.
This contribution is relevant for users, developers and theoreticians.
 For {\em users}, it gives an intuitive and classic view of the
 relational core patterns of SPARQL, allowing a good
understanding of how to use and combine the basic operators
of the SPARQL language when dealing with multisets.
For  {\em developers}, helps to perform optimization, design extensions of the
language, and understanding the semantics of multisets allowing for example
 translations from SPARQL operators to the right multiset  operators of SQL and vice versa.  
For {\em theoreticians}, introduces a clean framework (Multiset Datalog as
defined by Mumick et al. \cite{90820}) to study from a formal point of view
the multiset semantics of SPARQL patterns. 

The paper is organized as follows.
Section 2 presents the basic notions and notations used in the paper.
Section 3 identifies a classical relational algebra view of SPARQL patterns,
introducing the fragment $\core$.
Section 4 presents the equivalence between $\core$ and multiset non-recursive
Datalog with safe negation, and provides explicit transformations in both
directions.
Section 5 introduces the Multiset Relational Algebra, a simple and
intuitive fragment of  relational algebra with multiset semantics, and proves that it is
exactly equivalent to multiset non-recursive Datalog with safe
negation.
Section 6 analyzes related work and presents brief conclusions.

%=============================================
\section{SPARQL graph patterns}
\label{sec:patterns}

The definition of SPARQL graph patterns will be presented by using the formalism presented in \cite{10160}, but in agreement with the W3C specifications of SPARQL 1.0 \cite{10155} and SPARQL 1.1 \cite{90699}.

\paragraph{RDF graphs.}
Assume two disjoint infinite sets $I$ and $L$, called IRIs and literals respectively.\footnote{In addition to $I$ and $L$, RDF and SPARQL consider a domain of anonymous resources called blank nodes.  Their occurrence introduces  issues that are not discussed in this paper. 
 Based on the results in \cite{91033}, we avoided blank nodes assuming that their absence does not affect the results presented in this paper.} 
An \emph{RDF term} is an element in the set $T = I \cup L$.
 An \emph{RDF triple} is a tuple $(v_1,v_2,v_3) \in I \times I \times T$ where $v_1$ is the \emph{subject}, $v_2$ the
\emph{predicate} and $v_3$ the \emph{object}.  
An \emph{RDF Graph} (just graph from now on) is a set of RDF triples.  
The \emph{union} of graphs, $G_1 \cup G_2$, is the set theoretical
union of their sets of triples.  
 Additionally, assume the existence of an infinite set $V$ of variables disjoint from $T$.
We will use $\var(\alpha)$ to denote the set of variables occurring in the structure $\alpha$. 

A \emph{solution mapping} (or just \emph{mapping} from now on) is a partial function $\mu : V \to T$ where the domain of $\mu$, $\dom(\mu)$, is the subset of $V$ where $\mu$ is defined. 
The \emph{empty mapping}, denoted $\mu_0$, is the mapping satisfying that 
$\dom(\mu_0)= \emptyset$. 
Given $?X \in V$ and $c \in T$, we use $\mu(?X) = c$ to denote the solution mapping variable $?X$ to term $c$. 
Similarly, $\mu_{?X \to c}$ denotes a mapping $\mu$ satisfying that $\dom(\mu)=\{?X\}$ and $\mu(?X) = c$.
Given a finite set of variables $W \subset V$, the restriction of a mapping $\mu$ to $W$, denoted $\mu_{|W}$, is a mapping $\mu'$ satisfying that  $\dom(\mu') = \dom(\mu) \cap W$ and $\mu'(?X) = \mu(?X)$ for every $?X \in \dom(\mu) \cap W$.
 Two mappings $\mu_1, \mu_2$ are \emph{compatible}, denoted $\mu_1 \sim \mu_2$, when for all $?X \in \dom(\mu_1) \cap \dom(\mu_2)$ it satisfies that $\mu_1(?X)=\mu_2(?X)$, i.e., when $\mu_1 \cup \mu_2$ is also a mapping.
Note that two mappings with disjoint domains are always compatible, and that the empty mapping $\mu_0$ is compatible with any other mapping.

A \emph{selection formula} is defined recursively as follows: 
(i) If $?X,?Y \in V$ and $c \in I \cup L$ then $(?X = c)$, $(?X = ?Y)$ and $\bound(?X)$ are atomic selection formulas;
(ii) If $F$ and $F'$ are selection formulas then $(F \land F')$, $(F \lor F')$ and $\neg (F)$ are boolean selection formulas.
 The evaluation of a selection formula $F$ under a mapping $\mu$, denoted $\mu(F)$, is defined in a three-valued logic with values $\true$, $\false$ and $\error$. 
 We say that $\mu$ satisfies $F$ when $\mu(F) = \true$. 
 The semantics of $\mu(F)$ is defined as follows:

\begin{itemize}
\item If $F$ is $?X = c$ and $?X \in \dom(\mu)$, then $\mu(F) = \true$ when $\mu(?X) = c$ and $\mu(F) = \false$ otherwise. If $?X \notin \dom(\mu)$ then $\mu(F) = \error$.
\item If $F$ is $?X = ?Y$ and $?X,?Y \in \dom(\mu)$, then $\mu(F) = \true$ when $\mu(?X) = \mu(?Y)$ and $\mu(F) = \false$ otherwise. If either $?X \notin \dom(\mu)$ or $?Y \notin \dom(\mu)$ then $\mu(F) = \error$.
\item If $F$ is $\bound(?X)$ and $?X \in \dom(\mu)$ then $\mu(F) = \true$ else $\mu(F) = \false$.
\item If $F$ is a Boolean combination of the previous atomic cases,
then it is evaluated following a three value logic table (see \cite{10155}, 17.2).
\end{itemize}

\paragraph{Multisets.}
\label{p-multisets}
A \emph{multiset} is an unordered collection in which
each element may occur more than once.
A multiset $M$ will be represented as a set of pairs $(t,j)$, each pair denoting 
an element $t$ and the number $j$ of times
it occurs in the multiset (called multiplicity or cardinality).
When $(t,j) \in M$ we will say that $t$ $j$-belongs to $M$
(intuitively ``$t$ has $j$ copies in $M$'').
To uniformize the notation and capture the corner cases,
 we will write  $(t,*) \in M$ or simply say $t \in M$ when
there are $\geq 1$ copies of $t$ in $M$.
Similarly, when there is no occurrence of $t$ in $M$,
we will simply say  ``$t$ does not belong to $M$'', and 
abusing notation write $(t,0) \in M$, or
$(t,*) \notin M$. All of them indicate that $t$ does not occur in
$M$.

For \emph{multisets}  of solution mappings, following the notation
of SPARQL,  we will also use the symbol $\Omega$ to denote a multiset and $\cardi{\mu}{\Omega}$ to denote the cardinality of the mapping $\mu$ in the multiset $\Omega$. 
In this sense, we use $(\mu,n) \in \Omega$ to denote that $\cardi{\mu}{\Omega} = n$, or simply $\mu \in \Omega$ when $\cardi{\mu}{\Omega} > 0$.  
Similarly, $\cardi{\mu}{\Omega} = 0$ when $\mu \notin \Omega$.
The domain of a multiset $\Omega$ is defined as $\dom(\Omega) = \bigcup_{\mu \in \Omega} \dom(\mu)$.

%---------
\paragraph{SPARQL algebra.}
Let $\Omega_1,\Omega_2$ be multisets of mappings, $W$ be a set of variables and $F$ be a selection formula.
The \emph{SPARQL algebra for multisets of mappings} is composed of the
operations of projection, selection, join, union, minus, difference and left-join,  defined respectively as follows:
\begin{itemize}
\item 
$\pi_W(\Omega_1) = \{ \mu' \mid \exists \mu \in \Omega_1, \mu' = \mu_{|W} \}$\\ 
where 
$\cardi{\mu'}{\pi_W(\Omega_1)} = \sum_{\mu' = \mu_{|W}} \cardi{\mu}{\Omega_1}$

\item   
$\sigma_F(\Omega_1) = \{ \mu \in \Omega_1 \mid \mu(F) = \true \}$ \\ 
where $\cardi{\mu}{\sigma_F(\Omega_1)} = \cardi{\mu}{\Omega_1}$ 

\item 
$\Omega_1 \Join \Omega_2 = \{ \mu = (\mu_1 \cup \mu_2) \mid \mu_1 \in \Omega_1, \mu_2 \in \Omega_2, \mu_1 \sim \mu_2 \}$ \\
where
$\cardi{\mu}{\Omega_1 \Join \Omega_2} = \sum_{\mu = (\mu_1 \cup \mu_2)} \cardi{\mu_1}{\Omega_1}~\times~\cardi{\mu_2}{\Omega_2}$

\item 
$\Omega_1 \cup \Omega_2 = \{ \mu \mid \mu \in \Omega_1 \lor \mu \in \Omega_2 \}$ \\
where
$\cardi{\mu}{\Omega_1 \cup \Omega_2} = \cardi{\mu}{\Omega_1} + \cardi{\mu}{\Omega_2}$ 

\item 
$\Omega_1 - \Omega_2 = \{ \mu_1 \in \Omega_1 \mid \forall \mu_2 \in \Omega_2, \mu_1 \nsim \mu_2 \lor \dom(\mu_1) \cap \dom(\mu_2) = \emptyset \}$ \\
where
$\cardi{\mu_1}{\Omega_1 - \Omega_2} = \cardi{\mu_1}{\Omega_1}$ 

\item 
$\Omega_1 \setminus_F \Omega_2 = \{ \mu_1 \in \Omega_1 \mid \forall \mu_2 \in \Omega_2, (\mu_1 \nsim \mu_2) \lor (\mu_1 \sim \mu_2 \land (\mu_1 \cup \mu_2)(F) \neq \true ) \}$ \\
where
$\cardi{\mu_1}{\Omega_1 \setminus_F \Omega_2} = \cardi{\mu_1}{\Omega_1}$

\item
$\Omega_1 \lojoin_F \Omega_2 = \sigma_F(\Omega_1 \Join \Omega_2) \cup
(\Omega_1 \setminus_F \Omega_2)$ \\
where
$\cardi{\mu}{\Omega_1 \lojoin_F \Omega_2} =
\cardi{\mu}{\sigma_F(\Omega_1 \Join \Omega_2)} + \cardi{\mu}{\Omega_1
  \setminus_F \Omega_2}$ 

\end{itemize}

%---------
\paragraph{Syntax of graph patterns.}
A SPARQL \emph{graph pattern} is defined recursively as follows:
A triple from $(I \cup L \cup V) \times (I \cup V) \times (I \cup L \cup V)$ is a graph pattern called a \emph{triple pattern}.
\footnote{We assume that any triple pattern contains at least one variable.} 
If $P_1$ and $P_2$ are graph patterns then 
$( P_1 \AAND P_2 )$, 
$( P_1 \UNION P_2 )$, 
$( P_1 \OPTIONAL P_2 )$ and 
$( P_1 \MINUS P_2 )$ are graph patterns.
If $C$ is a filter constraint (as defined below) and
$var(C) \subseteq \dom(P_1)$, 
then $(P_1 \FILTER C)$ is a graph pattern.
And if $W \subseteq \dom(P_1)$ is a set of variables, then $(\SELECT W P_1)$ is a  graph pattern.

A \emph{filter constraint} is defined recursively as follows:
(i) If $?X,?Y \in V$ and $c \in I \cup L$ then $(?X = c)$, $(?X = ?Y)$ and $\bound(?X)$ are \emph{atomic filter constraints};
(ii) If $C_1$ and $C_2$ are filter constraints then 
$(!C_1)$, $(C_1~||~C_2)$ and $(C_1~\&\&~C_2)$ 
are \emph{complex filter constraints}.
Given a filter constraint $C$, we denote by $f(C)$ the selection formula
represented by $C$. Note that there exists a simple and direct translation from filter constraints to selection formulas and vice versa.

%---------
\paragraph{Semantics of SPARQL graph patterns.}
The evaluation of a SPARQL graph pattern $P$ over an RDF graph $G$ is defined as a function $\ev{P}{G}$ (or $\ev{P}{}$
when $G$ is clear from the context) which returns a multiset of solution mappings. 
 Let $P_1,P_2,P_3$ be graph patterns and $C$ be a filter constraint.
The evaluation of a graph pattern $P$ over a graph $G$
is defined recursively as follows:
\begin{enumerate}
\item If $P$ is a triple pattern $t$, then
$\ev{P}{G} = \{ (\mu,1) \mid \dom(\mu) = \var(t) \land \mu(t) \in G \}$
where
$\mu(t)$ is the triple obtained by replacing the variables in $t$ according to $\mu$.
%% and each mapping $\mu$ has cardinality 1.	
\item $\ev{(P_1 \AAND P_2)}{G} = \ev{P_1}{G} \Join \ev{P_2}{G}$.    
\item If $P$ is $(P_1 \OPTIONAL P_2)$ then
 \begin{itemize} 
 \item[(a)] if $P_2$ is $(P_3 \FILTER C)$ then $\ev{P}{G} = \ev{P_1}{G} \lojoin_C \ev{P_3}{G}$
 \item[(b)] else $\ev{P}{G} = \ev{P_1}{G} \lojoin_{(\true)} \ev{P_2}{G}$.  
 \end{itemize}
\item $\ev{(P_1 \MINUS P_2)}{G} =  \ev{P_1}{G} - \ev{P_2}{G}$.  
\item $\ev{(P_1 \UNION P_2)}{G} =  \ev{P_1}{G} \cup \ev{P_2}{G}$.  
\item $\ev{(P_1 \FILTER C)}{G} = \sigma_{f(C)}( \ev{P_1}{G})$. 
\item $\ev{(\SELECT  W P_1)}{G} = \pi_W(\ev{P_1}{G})$.
\end{enumerate}

For the rest of the paper, we will call $\wg$ the
fragment of graph patterns defined as follows:
\begin{definition}[$\wg$]
\label{wg}
  $\wg$ is the fragment of SPARQL
 composed of the operators $\AAND$, $\UNION$,
$\OPT$, $\FILTER$, $\MINUS$ and $\SELECT$, as defined above.  
\end{definition}

%=====================================================================================
%=====================================================================================

\section{The relational fragment of SPARQL}
\label{sec:core}

In this section we will introduce a fragment of SPARQL which follows
standard intuitions of the operators from relational algebra and SQL. We will prove that this fragment is equivalent to $\wg$.
 First, let us introduce the $\DIFF$ operator as an explicit way of expressing negation-by-failure\footnote{Recall that negation-by-failure can be expressed in SPARQL 1.0 as the combination of an optional graph pattern and a filter constraint containing the bound operator.} in SPARQL. 

\begin{definition}[The $\DIFF$ operator]
The weak difference of two graph patterns, $P_1$ and $P_2$, is defined as
\[ 
\ev{(P_1 \DIFF P_2)}{} = \{ \mu_1 \in \ev{P_1}{} \mid \forall \mu_2 \in \ev{P_2}{}, \mu_1 \nsim \mu_2 \}
\]
\end{definition}  
where $\cardi{\mu_1}{\ev{(P_1 \DIFF P_2)}{}} = \cardi{\mu_1}{\ev{P_1}{}}$.

It is  important to note that the DIFF operator is not defined in
SPARQL 1.0 nor in SPARQL 1.1 at the syntax level. However, it can be
implemented in  current SPARQL engines by using the  difference
operator of the $\wg$ algebra ($\Omega_1 \lojoin_{true} \Omega_2$).
 It was showed \cite{91301,91309} that the operators $\OPTIONAL$ and
 $\MINUS$ can be simulated 
with the operator $\DIFF$ in combination with $\AAND$, $\UNION$ and $\FILTER$.

In order to facilitate, and make more natural the translation from
SPARQL to Relational Algebra (and Datalog), we will introduce a more
intuitive notion of difference between two graph patterns.
We define the domain of a pattern $P$, denoted $\dom(P)$,
as the set of variables that occur (defining the output ``schema'') in  
the multiset of solution mappings for any evaluation of $P$.

\begin{definition}[The $\EXCEPT$ operator]
Let $P_1, P_2$ be graph patterns satisfying $\dom(P_1) = \dom(P_2)$.
The except difference of $P_1$ and $P_2$ is defined as
\[ 
\ev{(P_1 \EXCEPT P_2)}{} = \{ \mu \in \ev{P_1}{} \mid \mu \notin \ev{P_2}{} \},
\]
where $\cardi{\mu}{\ev{(P_1 \EXCEPT P_2)}{}} = \cardi{\mu}{\ev{P_1}{}}$.

We will denote by $\EXCEPT^*$ (or outer EXCEPT) the version of this operation when the restriction on domains is not considered.\footnote{This operation is called SetMinus in \cite{Kaminski}.} 
\end{definition} 

 Note that the restriction on the domains of $P_1$ and $P_2$ follows
 the philosophy of classical relational algebra. But it can be proved
that $\EXCEPT$ and its outer version are simulable each other:
\begin{lemma}
\label{outer}
In $\wg$,  the operator
$\EXCEPT$ can be simulated  using  $\EXCEPT^*$ and vice versa.
\end{lemma}

\begin{proof}
Clearly $\EXCEPT$ can be simulated by $\EXCEPT^*$.

On the other direction, let  $\dom(P) = X \cup Y$ and $\dom(Q) = X \cup Z$, 
where $X, Y, Z$ are disjoint set of variables. 
For a 
given set of variables $V = \{v_1 , \dots , v_n \}$, let 
NoneBound$(V )$ denotes $\neg \text{bound}(v_1) \wedge \cdots \wedge
\neg \text{bound}(v_n)$,
 and SomeBound$(V )$ denotes $\text{bound}(v_1) \vee \cdots \vee 
 \text{bound}(v_n)$. Let $P'$ be $(P \FILTER \text{NoneBound}(Y))$ 
and $Q'$ be $(Q \FILTER \text{NoneBound}(Z ))$. 
 Let $P ''$ and $Q''$ be the graph patterns 
$(\SELECT X P')$ and $(\SELECT X Q')$ respectively. 
Now, $\dom(P'') = \dom(Q'') = X$ and hence
 $(P'' \EXCEPT Q'')$ makes sense. Thus
\[
(P \EXCEPT^* Q) \equiv ((P '' \EXCEPT Q'' ) \UNION (P \FILTER
\text{SomeBound}(Y ))).
\]
\end{proof}

 The next lemma
establishes the relationship between $\EXCEPT$ and $\DIFF$,
showing that $\EXCEPT$ can be simulated in $\wg$.
 
\begin{lemma}
\label{lemma:diffex}
For every pair of graph patterns $P_1,P_2$ in $\wg$, and any RDF graph
$G$,  the operator $\EXCEPT$ can be simulated by $\DIFF$ and vice versa.
\end{lemma} 

\begin{proof}
The high level proof goes as follows. 
As we saw before, $\EXCEPT$ and
$\EXCEPT^*$ are mutually simulable in $\wg$. And $\EXCEPT^*$ differs from DIFF only in checking compatibility of mappings (i.e. $\sim$).
$\ev{P_1 \EXCEPT^* P_2}{}$ eliminates from
$\ev{P_1}{}$ those mappings in $\ev{P_2}{}$ that are equal to one in 
$\ev{P_1}{}$;
while $\DIFF$ eliminates those that are compatible with one in $\ev{P_1}{}$. 
 That is, the difference is between the multisets
$\{ (\mu_1,n_1)\in \Omega_1 \mid \neg\exists \mu_2 \in \Omega_2 ~\land~ \mu_1 = \mu_2   \}$ versus
$\{ (\mu_1,n_1) \in \Omega_1 \mid \neg\exists  \mu_2 \in \Omega_2 ~\land~ \mu_1 \sim \mu_2 \}$.
Now, for two mappings $\mu_1,\mu_2$, equality and compatibility 
($\mu_1=\mu_2$ versus $\mu_1 \sim \mu_2$) differ only in those
variables that are bound in $\mu_1$ and unbound in $\mu_2$ or vice versa.
Thus, to simulate $=$ with $\sim$ and vice versa, 
it is enough to have an operator that replaces all unbound 
entries in mappings of  $\Omega_1$ and
$\Omega_2$ by a fresh new constant, e.g. $c$, call the new sets 
$\Omega_1'$ and $\Omega_2'$, and we will have that 
$\{ (\mu_1,n_1) \in \Omega_1 \mid \neg\exists \mu_2 \in \Omega_2 ~\land~ \mu_1 \sim \mu_2 \}$ is equivalent to 
$\{ (\mu_1,n_1) \in \Omega_1' \mid \neg\exists  \mu_2 \in \Omega_2' ~\land~ \mu_1 = \mu_2 \}$. 
Note that cardinalities are preserved because the change between
 ``unbound'' and ``c'' does not change them.
The rest is to express the  two operations on 
multisets of solution mappings:  the one that fills in unbound entries with a fresh
constant $c$; 
and the one that changes back the values $c$ to unbound.
\end{proof}

 With the new operator $\EXCEPT$ we define the following relational
fragment of SPARQL:
\begin{definition}
 Define $\core$ as the fragment of $\wg$ graph pattern expressions
 defined recursively by triple patterns plus the operators $\AAND$,
 $\UNION$,  $\EXCEPT$, $\FILTER$ and $\SELECT$.
\end{definition}

Now we are ready to state the main theorem.
Considering that $\DIFF$ is able to express $\OPT$ and $\MINUS$ (cf. \cite{91301,91309}), and that the $\DIFF$ operator is expressible in
$\core$ (Lemma \ref{lemma:diffex}), we have the following result:

\begin{theorem}
\label{theo:core=wg}
$\core$ is equivalent to $\wg$.
\end{theorem} 

For the rest of the paper, we will concentrate our interest on $\core$.  

\borrado{
\begin{note}
An alternative proof of Theorem \ref{theo:core=wg} is given as follows.
(Compare~\cite{91309}, Lemma 12).
Let $\theta$ be a function that renames variables by fresh ones.

$\wg$ contains $\core$: 
The graph pattern $(P_1 \EXCEPT P_2)$ can be rewritten into an equivalent pattern 
$(((P_1 \OPT (\theta P_2)) \FILTER C) \FILTER C')$ 
where
$\dom(P_1) = \{ ?x_1, \dots, ?x_n \}$,
$C$ is $( ?x_1 = \theta ?x_1 ~ \&\& \dots \&\& ~ ?x_n = \theta ?x_n)$
and
$C'$ is $(!\bound(\theta ?x_1))$.

$\core$ contains $\wg$: 
The graph pattern $(P_1 \DIFF P_2)$ can be rewritten into an equivalent graph pattern

\hspace{0.5cm} $(P_1 \EXCEPT (\SELECT W ((P_1 \AAND P_1') \FILTER C) \AAND P_2'))$ \\
where 
$W = \dom(P_1) = \{ ?x_1, \dots, ?x_n \}$, 
$P_1' = \theta(P_1)$, 
$P_2' = \theta(P_2)$ 
and 
$C$ is $( ?x_1 = \theta ?x_1 ~ \&\& \dots ~ ?x_n = \theta ?x_n)$.
\end{note}
}

%=====================================================================================
%=====================================================================================

\section{$\core$ $\equiv$ Multiset Datalog}
\label{sec:core=datalog}
 In this section we prove that $\core$ have the same expressive power
of Multiset Datalog. Although the ideas of the proof are similar to
those in \cite{90006} (now for $\core$), we will sketch the main transformations
to make the paper as self contained as possible. For notions of
Datalog see  Levene and Loizou \cite{70034},
for the semantics of Multiset Datalog, Mumick et al. \cite{90820}.

\subsection{Multiset Datalog}
\label{sec:datalog}
A \emph{term} is either a variable or a constant.
A positive literal $L$ is either a \emph{predicate formula}
$p(t_1,\dots, t_n)$ where $p$ is a predicate name and $t_1,\dots, t_n$ are terms, or an \emph{equality formula} $t_1 = t_2$ where $t_1$ and $t_2$ are terms. A negative literal $\neg L$ is the negation
of a positive literal $L$.
A \emph{rule} is an expression of the form
$L \gets L_1 \land \dots \land L_k \land \neg L_{k+1} \land \dots \land \neg L_n$ 
where $L$ is a positive literal called the \emph{head} of the rule and the rest of literals (positive and negative) are called the \emph{body}.
A \emph{fact} is a rule with empty body and no variables.
A \emph{Datalog program} $\Pi$ is a finite set of rules and its set of facts is denoted $\facts(\Pi)$.

A variable $x$ is \emph{safe} in a rule $r$ if it occurs in a positive
predicate or in $x = c$ ($c$ constant) or in 
 $x = y$ where $y$ is safe. A rule is safe it all its variables are safe.
A program is \emph{safe} if all its rules are safe.
A program is non-recursive if its dependency graph is acyclic.
In what follows, we only consider non-recursive and safe Datalog programs,
denoted by nr-Datalog$^\neg$.

To incorporate multisets to 
the classical Datalog framework we will follow the approach introduced 
 by Mumick and Shmueli~\cite{50675}. 
 The idea is rather intuitive:
Each  derivation tree gives rise to a substitution $\theta$.
 In the standard (set) semantics, what  matters is the set
of the different substitutions that instantiates the distinguished literal. 
On the contrary, in multiset semantics the number of such instantiations 
also becomes relevant. As Mumick and Shmueli state~\cite{50675,90820},
``duplicate semantics of a program is obtained by counting the
number of derivation trees''.
  Thus now we have pairs $(\theta,n)$ of substitutions $\theta$
plus the number $n$ of derivation trees that produce $\theta$.

A Datalog query is a pair $(\Pi, L)$ where $\Pi$ is a program and $L$ is a
distinguished predicate (the goal) occurring as the head of a rule. 
The answer to $(\Pi,L)$ is the multiset of substitutions $\theta$ such
that makes $\theta(L)$ true.

\paragraph{Normalized Datalog.}
It is possible to have each  non-recursive Datalog with safe negation
 program written in a normalized form as the following lemma shows:

\begin{lemma}
\label{normalized}
  Each  Datalog program $P$ is equivalent to a program $P'$
that only uses  the following types of rules, where
$L, L_1, L_2$ be predicate formulas, and $EQ$ is a set of equality and
inequality formulas:
\begin{itemize}
\item (Projection rule) $L \leftarrow L_1$ where $\var(L) \subset \var(L_1)$;
\item (Selection rule) $L \leftarrow L_1, EQ$, where $\var(L) =
  \var(L_1) \cup \var(EQ)$ and the rule is safe;
\item (Join rule) $L \leftarrow L_1, L_2$,
 where $\var(L) = \var(L_1) \cup \var(L_2)$; and 
\item (Negation rule) $L \leftarrow L_1, \neg L_2$ where 
$\var(L_2)  \subseteq  \var(L_1)$ and $\var(L) = \var(L_1)$. 
\end{itemize}
\end{lemma}
 
\begin{proof}
The idea of the proof (i.e. it shows expressive equivalence, but
has no efficiency considerations) is as follows.
The general rule of $P$ has the form:
\begin{equation}
\label{rule}
L \leftarrow L_1, \dots, L_n, \neg H_1,\dots, \neg H_m, EQ,  
 \end{equation}
where $L_i$ and $H_j$ are predicate formulas, and $EQ$ is a set 
of equality and inequality formulas, and the rule is safe.

First we may assume that $n=1$, i.e. that there is only one positive
literal $L_1$.
For this,  note that $L_{13}\leftarrow L_1, L_2, L_3$, where 
$L_{13}$ is a fresh predicate formula whose variables are exactly
those in $L_1, L_2, L_3$,   can be rewritten in two rules
$L_{13} \leftarrow L_{12}, L_3$  and $L_{12} \leftarrow L_1, L_2$,
with the intended meanings.
Then proceed recursively with the other positive literals and we 
can assume that there is a unique literal $L_{1n}$ defined by
the above rules.
  
Second, we may assume that $EQ$ contains only inequalities, as
we can get rid of the equalities as follows: 
(a) for $x = y$ both variables, replace everywhere in the rule
 $x$ by $y$; (b) for $x = c$, where $c$ is constant, replace
everywhere in the rule $x$ by $c$; if $c = c$, just eliminate it;
and if $a \neq b$ different constants, then replace the whole
rule by $L \leftarrow a \neq b$. 
  Let us denote $EQ'$ the remaining set of inequalities. 

Now, because   all variables in $EQ'$ must be in $L_{1n}$ (because the rule
 (\ref{rule}) is safe), the rule $L' \leftarrow L_{1n}, EQ'$,
where $L'$ is a fresh predicate formula containing all the variables 
in the body, is well  defined and safe.

At this point we have reduced the rule (\ref{rule}) to
 $L \leftarrow L', \neg H_1,\dots, \neg H_m$.
 Now define recursively 
$L'' \leftarrow L', \neg H_2$  were $L''$ 
is a fresh predicate formula containing all the variables in the body, 
and so on, until we get the rule $L \leftarrow L^{(m-1)}, H_m$.

The desired program $P'$ is the set of all the rules so defined.
\end{proof}

%--------------------------------------------
\subsection{From SPARQL to Datalog}
\label{sec:sparql2datalog}
The algorithm that transforms  SPARQL into Datalog
includes transformations of RDF graphs to Datalog facts,
SPARQL  queries  into a Datalog queries, and
SPARQL mappings into Datalog substitutions.

%-------------
\paragraph{RDF graphs to Datalog facts.} 
Let $G$ be an RDF graph:
each term $t$ in $G$ is encoded by a fact $iri(t)$ or $literal(t)$ when $t$ is an IRI
or a literal respectively;
the set of terms in $G$ is defined by the rules $term(X) \gets iri(X)$ and $term(X) \gets literal(X)$;
the fact $Null(null)$ encodes the \emph{null} value  (unbounded value);
each RDF triple $(v_1,v_2,v_3)$ in  $G$ is encoded by a fact $triple(v_1,v_2,v_3)$.
Recall that we are assuming that an RDF graph is a ``set'' of triples.

%-------------
\paragraph{SPARQL patterns into Datalog rules:}
The transformation follows essentially the idea presented by
Polleres~\cite{10150}. Let $P$ be a graph pattern and $G$ an RDF
graph. Denote by $\delta(P)_G$ 
the function which transforms $P$ into a set of Datalog rules.
Table~\ref{table:pattern2rules} shows the transformation rules defined
by the function $\delta(P)_G$, where the notion of compatible mappings is implemented by the rules:

$comp(X,X,X) \gets term(X)$, $comp(X,Y,X) \gets term(X) \land Null(Y)$,

$comp(Y,X,X) \gets Null(Y) \land term(X)$, $comp(X,X,X) \gets Null(X)$.

Also, an atomic filter condition $C$ is encoded by a literal $L$ as follows
(where $?X,?Y \in V$ and $u \in I \cup L$):
if $C$ is either $(?X = u)$ or $(?X=?Y)$ then $L$ is $C$;
if $C$ is $\bound(?X)$ then $L$ is $\neg Null(?X)$.

\paragraph{SPARQL mappings to Datalog substitutions:}
Let $P$ be a graph pattern, $G$ an RDF graph and $\mu$ a solution
mapping of $P$ in $G$. 
Then $\mu$ gets transformed into 
a substitution $\theta$  satisfying that
for each $x \in \var(P)$ there exists $ x / t \in \theta$ 
such that $t = \mu(x)$ when $\mu(x)$ is bounded and $t = null$ otherwise.

Now, the correspondence between the multiplicities 
of mappings and substitutions works as follows:
Each SPARQL mapping comes from an evaluation tree. 
A {\em set} of evaluation trees becomes a {\em multiset} of mappings.
Similarly, a {\em set} of Datalog derivation trees becomes 
a {\em multiset} of substitutions.
Thus, each occurrence of a mapping $\mu$ comes from a SPARQL evaluation tree. 
This tree is translated by Table~\ref{table:pattern2rules} to a 
Datalog derivation tree, giving rise to an occurrence of a substitution
in Datalog.
Each recursive step in Table~\ref{table:pattern2rules} carries out 
bottom up the correspondence between cardinalities of mappings and substitutions.

Note that in Table~\ref{table:pattern2rules}, the translation for
filters only consider conditions $C$ atomic. This is justified by the 
following Lemma:

\begin{lemma}
\label{conditionC}
Given a pattern $P$, for each Boolean formula $C$ in $\core$, where
$C$ is a general Boolean condition (Boolean combination of atomic
terms), there is pattern $P'$ in $\core$ that uses only atomic filters
(i.e. conditions of the form $(t_1=t_2)$ or $(t_1 \neq t_2)$, where $t_1,t_2$ are
variables or constants),  such that for each $G$,
$\ev{P \FILTER (C)}{G} = \ev{P'}{G}$.
\end{lemma}

\begin{proof}
First, consider the conjunctive normal form of $C$, namely 
$\bigwedge_{j=1}^{m} D_j$, where each $D_j$ is a disjunction of 
equalities or negation of equalities in a set $E$.
Then\\
$\ev{P \FILTER C}{G} = \ev{(\cdots(P \FILTER D_1) \FILTER D_2)\cdots
  )\FILTER(D_m))}{G}$.

Thus we can asume that $C$ is a disjunction $D = (d_1 \vee \dots \vee
d_k)$ where each $d_j$ is an equality or a negation of an equality in $E$.
We will show the case $k=2$ and it is not difficult to see how to
generalize it. We have the logical equivalence:
\begin{equation}
\label{dd}
d_1 \vee d_2 \equiv  (d_1 \wedge  \neg d_2) \vee  (\neg d_1 \wedge  d_2) \vee (d_1
  \wedge d_2).
\end{equation}
We claim that with multiset semantics
\begin{multline}
\label{m}
\ev{P \FILTER (d_1 \vee d_2 )}{G} =  \ev{(P \FILTER (d_1 \wedge  \neg
  d_2)) \UNION  \\ (P \FILTER  (\neg d_1 \wedge  d_2)) 
\UNION  (P \FILTER    (d_1  \wedge d_2)  )   }{G}.
\end{multline}
Now, the crucial point is to observe that each mapping in
the left-hand side of (\ref{m})
satisfies {\em one and only one} of the term of the union in the right-hand side.
Hence the equivalence preserves multiplicity of mappings.

Now use again the fact that each 
$(P \FILTER  (d_i  \wedge d_j))$ is equivalent to 
 $(P \FILTER  (d_i)) \FILTER(d_j)$, and we get the desired pattern
 $P'$ having only atomic filter conditions.
\end{proof}

\begin{table}[t!]
\caption{
Transforming $\core$ graph patterns into Datalog Rules. 
The function $\delta(P)_G$ takes a graph pattern $P$ and an RDF graph
$G$, and returns a set of Datalog rules with main predicate
$p(\overline{\var}(P))$, where
$\overline{\var}(P)$ denotes the tuple of variables obtained 
from a lexicographical ordering of the variables in $P$.
If $L$ is a Datalog literal, then $\nu_j(L)$ denotes a copy of $L$ with its variables renamed according to a variable renaming function $\nu_j : V \to V$.
$comp$ is a literal encoding the notion of compatible mappings. 
$cond$ is a literal encoding the filter condition $C$.
$\overline{W}$ is a subset of $\overline{\var}(P_1)$.
}
\footnotesize
\centering
\begin{tabular}{|l|l|}
 \hline

 \rs Pattern $P$        
 & $\delta(P)_G$ \\ \hline \hline 

 \rs $(x_1,x_2,x_3)$    
 & $p( \overline{\var}(P) ) \gets triple(x_1,x_2,x_3)$  \\ \hline 

 \rs $(P_1 \AAND P_2)$   
 & $p( \overline{\var}(P)) \gets \nu_1(p_1(\overline{\var}(P_1)))~\land~\nu_2(p_2(\overline{\var}(P_2)))$ \\
 & $\hspace{4cm}\bigwedge_{x \in \var(P_1) \cap \var(P_2)} comp(\nu_1(x),\nu_2(x),x)$, \\
 & $\delta(P_1)_G$ , $\delta(P_2)_G$ \\  \cline{2-2}
 & \rs $\dom(\nu_1) = \dom(\nu_2) = \var(P_1) \cap \var(P_2)$, $\range(\nu_1) \cap \range(\nu_2) = \emptyset$. \\ 
 \hline

 \rs $(P_1 \UNION P_2)$ 
 & $p(\overline{\var}(P)) \gets p_1(\overline{\var}(P_1)) \bigwedge_{x \in \var(P_2) \setminus \var(P_1)}Null(x)$, 
 \\
 & \rs $p(\overline{\var}(P)) \gets p_2(\overline{\var}(P_2)) \bigwedge_{x \in \var(P_1) \setminus \var(P_2)}Null(x)$,
 \\ 
 & $\delta(P_1)_G$ , $\delta(P_2)_G$ \\
 \hline

 \rs $(P_1 \EXCEPT P_2)$ 
 & $p(\overline{\var}(P_1)) \gets p_1(\overline{\var}(P_1)) \land \neg p_2(\overline{\var}(P_2))$, \\
 & $\delta(P_1)_G$ , $\delta(P_2)_G$ \\
 \hline

\rs $(\SELECT W P_1)$ 
 & $p(\overline{W}) \gets p_1(\overline{\var}(P_1))$, \\
 & $\delta(P_1)_G$ \\
 \hline

 \rs $(P_1 \FILTER C)$ 
 & $p( \overline{\var}(P)) \gets p_1(\overline{\var}(P_1)) \land cond$ \\
 and \rs $C$ is atomic
 & $\delta(P_1)_G$ \\ 
 \hline  
\end{tabular}
\normalsize
\label{table:pattern2rules}
\end{table}

Thus we have that a SPARQL query $Q = (P,G)$ where
$P$ is a graph pattern and $G$ is an RDF graph gets
transformed into the Datalog query
$(\Pi,p(\overline{\var}(P)))$ where
$\Pi$ is the Datalog program $\delta(P)_G$ plus the facts
got from the transformation of the graph $G$,
and $p$ is the goal literal related to $P$.

%-----------------------------------
\subsection{From Datalog to SPARQL}
\label{sec:datalog2sparql}
Now we need to transform Datalog facts into RDF data, Datalog substitutions into  SPARQL mappings, and Datalog queries into SPARQL queries.

\paragraph{Datalog facts as an RDF Graph:}
Given a Datalog fact $f = p(c_1,...,c_n)$, consider the function $\desc(f)$ which returns the set of triples

\hspace{1cm}$ \{ (u,\text{predicate},p),
(u,\text{rdf:\_1},c_1),\dots,(u,\text{rdf:\_n},c_n) \},$ \\
where $u$ is a fresh IRI.
Given a set of Datalog facts $F$, the RDF description of $F$ will be the graph $G = \bigcup_{f \in F} \desc(f)$.

\paragraph{Datalog rules as SPARQL graph patterns:}
Let $\Pi$ be a (normalized) Datalog program and $L$ be a literal $p(x_1,\dots,x_n)$ where $p$ is a predicate in $\Pi$ and each $x_i$ is a variable.
We define the function $\gp(L)_{\Pi}$ which returns a graph pattern encoding of the program $(\Pi, L)$.
 The translation works intuitively as follows: 
\begin{itemize}
\item[(a)] 
If predicate $p$ is extensional, then $\gp( L )_{\Pi}$ returns the graph pattern\\
$( (?Y, \text{predicate},p) \AAND ( ?Y, \text{rdf:\_1},x_1) \AAND \cdots \AAND (?Y, \text{rdf\_n}, x_n) )$,\\
where $?Y$ is a fresh variable. 

\item[(b)] 
If predicate $p$ is intensional and $\{r_1,\dots, r_n\}$ is the set of all the rules in $\Pi$ where $p$ occurs in the head, then $\gp(L)_{\Pi}$ returns the graph pattern 
$(\dots(  T(r_1) \UNION T(r_2) ) \dots \UNION T(r_n))$ where $T(r_i)$ is defined as follows (when $n = 1$ the resulting graph pattern is reduced to $T(r_1)$):
\begin{itemize}

\item If $r_i$ is $L \leftarrow L_1$ then $T(r_i)$ returns \\
$\SELECT x_1,\dots,x_n \WHERE \gp(L_1)_{\Pi}$.

\item If $r_i$ is $L \leftarrow L_1 \land EQ$, 
where $EQ$ is a set of equalities or negations of equalities, then 
$T(r_i)$ returns
$(\gp(L_1)_{\Pi} \FILTER C)$ where $C$ is a filter condition equivalent to $EQ$.

\item If $r_i$ is $L \leftarrow L_1 \land L_2$ then $T(r_i)$ returns
$(\gp(L_1)_{\Pi} \AAND \gp(L_2)_{\Pi})$.

\item If $r_i$ is $L \leftarrow L_1 \land \neg L_2$ then $T(r_i)$ returns\\
$(\gp(L_1)_{\Pi} \EXCEPT^* \gp(L_2)_{\Pi})$.
\end{itemize}
\end{itemize}

\paragraph{Datalog substitutions as SPARQL mappings:}
For each substitution $\theta$ satisfying $(\Pi,L)$ build a mapping $\mu$
satisfying that,  if $x / t \in \theta$ then $x \in \dom(\mu)$ and $\mu(x) = t$.
The correspondence of multiplicities work in a similar way (via
derivation tree to evaluation tree) as in the case of mappings to substitutions.

\smallskip

 Putting together the transformation in Table
 \ref{table:pattern2rules} and the
pattern obtained by using $\gp( L )_{\Pi}$, we get the following theorem, whose proof is a 
long but straightforward induction on the structure 
of the patterns in one direction, and on  the level of Datalog in
the other.

\begin{theorem}
\label{theorem:datalog2sparql}
Multiset nr-Datalog$^\neg$ has the same expressive power as $\core$.
\end{theorem}

%=====================================================================================
%=====================================================================================

%%%%%%%%%%%%%%%%%%%%%%%%%%%%%%%%%%%%%%%%%%%
\section{The relational version of Multiset Datalog: MRA}

In this section we introduce a multiset relational algebra (called
MRA), counterpart of  Multiset Datalog, and prove its equivalence with
the fragment of non-recursive Datalog with safe negation.

%---------------------------
\subsection{Multiset Relational Algebra (MRA)}

Multiset relational algebra is an extension of classical relation
algebra having multisets of relations instead of sets of relations.
As indicated in the introduction, there are manifold approaches and
operators to extend set relational algebra with multisets.
We use the semantics of multiset operators defined 
by Dayal et al.~\cite{91030} for the operations
of  selection, projection, natural join and  arithmetic union;
and add filter difference (not present there) represented by
the operator ``except''.

 Let us formalize these notions.
In classical {\em (Set) relational algebra}, a database schema is a set of 
relational schemas.
A relational schema is defined as a set of attributes.
Each attribute $A$ has a domain, denoted $\dom(A)$.
A {\em relation} $R$ over the relational schema $S = \{A_1,\dots, A_n\}$
is a finite {\em set} of tuples.
 An instance $r$ of a schema $S$ is a relation over $S$.
Given an instance $r$ of a relation $R$ with schema $S$, 
$A_j \in S$ and $t =(a_1,\dots,a_n) \in r$, we denote by $t[A_j]$ the 
tuple $(a_j)$. Similarly with $t[X]$ when $X \subseteq S$ and 
we will define $t[\emptyset] = \emptyset$.

In the {\em Multiset relational algebra} setting, an instance of a schema is
a {\em multiset relation}, that is,
a set of pairs $(t,i)$, where $t$ is a tuple over the schema $S$,
and $i\geq 1$ is a positive integer. 
(For notions and notations on multisets recall section
\ref{p-multisets}, {\em Multisets}).

\begin{definition}[Multiset Relational Algebra (MRA)]
Let $r$ and $r'$ be multiset relations over the schemas $S$ and $S'$
respectively.
Let $A,B \in S$ be  attributes, $a \in \dom(A)$ and $I = S \cap S'$.
MRA consists of the following operations:

\begin{enumerate}
\item Selection. $\sigma_{C}(r) = \{ (t,i): (t,i) \in r \wedge t[A]=a
  \}$, where $C$ is a Boolean combination of terms of the form
 $A = B$ or $A=a$

\item Natural Join. $r \bowtie r'$ is a multiset relation over $S \cup S'$ defined as follows.
Let $S'' = S' - S$.
 Let  ${t}^{\smallfrown}t'$ denotes concatenation of tuples. Then
\[
r \bowtie r' = \{  (t^{\smallfrown}(t'[S'']), i \times j):  (t,i) \in r \wedge (t',j)\in r'  \wedge t[I] = t'[I]      \}.
\]

\item Projection. Let $X \subseteq S$. Then:
\[
 \pi_{X}(r) = \{  (t[X], \mbox{$\sum_{(t_j,n_j)\in r \mbox{ s.t. } t_j[X] = t}$}\; n_j ) : (t,*)\in r \}.
\]

\item Union. Assume $S = S'$.  
\begin{align*}
  r \uplus r' =& \{ (t,i):  \text{ $t$ $i$-belongs to $r$ and $t \notin r'$ } \} \\
            &  \cup  \{ (t',j):  \text{  $t' \notin r$ and $t'$ $j$-belongs to $r'$ } \} \\
            & \cup \{ (t, i+j): \text{ $t$ $i$-belongs to $r$ and $t$ $j$-belongs to $r'$ } \}.
\end{align*}

\item Except.  Assume $S = S'$. 
\[
  r \setminus r' = \{  (t,i) \in r :  (t,*) \notin r' \}.
\]
 \end{enumerate}
 As usual, we will define a {\em query} in this multiset relational algebra as
an expression over an extended domain which includes, besides the
original domains of the schemas, a set of variables $V$.
\end{definition}

%---------------------------------------
\subsection{MRA $\equiv$ Multiset nr-Datalog$^\neg$}
This subsection is devoted to prove the following result.

\begin{theorem}
Multiset relational algebra ({\bf MRA}) has the same expressive power 
as Multiset Non-recursive Datalog with safe negation.
\end{theorem}

From this theorem and Theorem \ref{theorem:datalog2sparql} it follows:

\begin{corollary}
$\core$ is equivalent to MRA.
\end{corollary}

\begin{proof}[of the Theorem]
The proof is based on the ideas of the proof of 
Theorem 3.18 in \cite{70034}, extended to multisets.

Let $E$ be a relational algebra query expression over the schema $R$ 
and $D$ a database. We may assume, using similar arguments as in
Lemma \ref{conditionC}, that the condition $C$ in in the select operator
is a conjunctions of equalities and inequalities of terms.
Then it will be translated by a function $(\cdot)^{\Pi}$
to the Datalog program $(\facts(\Pi) \cup E^{\Pi}, out_E)$,
where $\facts(\Pi)$ is the multiset of facts (over fresh predicates
$r^{\Pi}$ for each relation $r$, and having the same arity as the original schema of $r$):\\
$\facts(\Pi) = \{ (r^{\Pi}(t),n) : \text{ $t$ is a tuple with multiplicity $n$ in schema
  $r$ in $D$ } \}$,

\noindent
and $(E^{\Pi},out_E)$ is the datalog program produced by the
translation of the expression $E$ given by the recursive specification  below.
For the expression $E_j$, the set  $V_j$ will denote its list of 
attributes.

\begin{enumerate} 

\item  Base case. No operator involved. Thus the query is a member of the
schema $R$, namely $r(x_1,\dots,x_n)$. The corresponding Multiset Datalog
query is:   

$out_r(x_1,\dots,x_n) \leftarrow r^{\Pi}(x_1,\dots, x_n)$

\item $E = \sigma_C(E_1)$, where $C$ is a conjunction of equalities 
and inequalities of terms.
The  translation is the program $(E^{\Pi}, out_E)$ where $E^{\Pi}$ is the 
program: $E_1^{\Pi} \cup \{R\}$ where $R$ is the rule: 

$out_E(x_1,\dots,x_k) \leftarrow out_{E_1}(x_1,\dots,x_k) \land C$.

\item  $E = E_1 \bowtie E_2$. Let $V = V_2 \setminus V_1$. 
The  translation is the program $(E^{\Pi}, out_E)$ where $E^{\Pi}$ is the 
program: $E_1^{\Pi} \cup E_2^{\Pi}  \cup \{R\}$ where $R$ is the rule: 

$out_E(V_1,V) \leftarrow out_{E_1}(V_1) \land  out_{E_2}(V_2)$.

\item  $E = \pi_{A}(E_1)$, where $A$ is a sublist of  the attributes
  in $E_1$.
The  translation is the program $(E^{\Pi}, out_E)$ where $E^{\Pi}$ is the 
program: $E_1^{\Pi} \cup \{R\}$ where $R$ is the rule: 

$out_E(A) \leftarrow out_{E_1}(V_1)$.

\item $E = E_1 \cup E_2$, where $E_1$ and $E_2$ have the same schema.
The  translation is the program $(E^{\Pi}, out_E)$ where $E^{\Pi}$ is the 
program: $E_1^{\Pi} \cup E_2^{\Pi}  \cup \{R1, R_2\}$ where $R1, R_2$
are  the rules: 

$out_E(x_1,\dots,x_k) \leftarrow out_{E_1}(x_1,\dots,x_k)$,

$out_E(x_1,\dots,x_k) \leftarrow out_{E_2}(x_1,\dots,x_k)$.

\item $E = E_1 \setminus E_2$,
  where  $E_1$ and $E_2$ have the same schema.
The  translation is the program $(E^{\Pi}, out_E)$ where $E^{\Pi}$ is the 
program: $E_1^{\Pi} \cup E_2^{\Pi}  \cup \{R1, R_2\}$ where $R1, R_2$
are  the rules: 

$out_E(x_1,\dots,x_k) \leftarrow out_{E_1}(x_1,\dots,x_k) \land \neg  out_{E_2}(x_1,\dots,x_k)$.
\end{enumerate}

It is important to check that the resulting program is non-recursive
(this is because the structure of the algebraic relational expression from
where it comes is a tree). Also it is safe because in rule (6) both
expressions have the same schema).
Now, it needs to be shown that 
for each relational expression (query) 
$E$ in $R$, $[E]_D$ and $[E^{\Pi}]$ return the 
same ``tuples'' with the same multiplicity.
 This is done by induction on the structure of $E$.

\smallskip

 Now, let us present the transformation from Multiset Datalog to 
Multiset Relational Algebra.
Note that we may assume a normal form for the Datalog programs as
shown in Lemma \ref{normalized}.
Then the recursive translation $(\cdot)^R$ from 
Datalog programs to MRA expressions  goes as follows.

\begin{enumerate}

\item First translate those head predicates $q$ occurring 
in $\geq 2$ rules as follows. 
 Let $q$ be the head of the rules $r_1,\dots, r_k$, $k \geq 2$. 
Rename each such head $q$ with the same set of variables $V$.
Then the translation is $(q)^R = (q_{r_1})^R \cup \cdots \cup
(q_{r_k})^R$.

\smallskip

From now on, we can assume that, not considering these $q$'s,
all other predicates occur as head in at most one rule. Hence we will
not need the subindex indicating the rule to which they belong to.

\item (Base case.) Let $r$ be a fact $q(V).$
 Then translates it
as  $(q_r)^R = q^R(V)$, where $q^R$ is a fresh new schema with the 
corresponding arity.

\item 
Let $r$ be
$q(A) \leftarrow p(V)$, where $A$ is a sublist of $V$.
The translation is  
  $(q_r)^R = \pi_{A}((p)^R)$.

\item
Let $r$ be
$q(V) \leftarrow p(V) \land C$,
where $C$ is a conjunction of equalities or inequalites of terms.
The translation is $(q_r)^R = \sigma_C((p)^R)$.

\item Let $r$ be
 $q(X,Y,Z) 
       \leftarrow p_1(X,Y) \land p_2(Y,Z)$,
where $X,Y,Z$ are disjoint lists of variables.
The translation is $(q_r)^R = (p_1)^R \bowtie (p_2)^R$.

\item 
Let $r$ be
$q(X,Y) \leftarrow p_1(X,Y) \land \neg  p_2(Y)$, that is the rule is safe.
The translation is 
$(q_r)^R = (p_1)^R \setminus   ((p_1)^R \bowtie  (p_2)^R)$.

\end{enumerate}

The arguments about multiplicity are straightforward verifications.
And because the program $\Pi$ is non-recursive (i.e. its dependency graph is acyclic), the recursive translation to the  relational
expression gives a well formed algebraic expression.
\end{proof}

%=====================================================================================
%=====================================================================================

\section{Related Work and Conclusions}

To the best of our knowledge, 
the multiset semantics of SPARQL has not been systematically addressed.
There are works that, when studying the expressive power of SPARQL,
 touched some aspects of this topic.
Cyganiak \cite{10140} was among the first who
 gave a translation of a core fragment of SPARQL into
relational algebra.  
Polleres \cite{10150} proved the inclusion of the fragment of
SPARQL patterns with safe filters into Datalog by giving a precise and
correct set of rules. 
Schenk \cite{10158} proposed a formal semantics for SPARQL based on 
Datalog, but concentrated on complexity more than expressiveness
issues. Both, Polleres and Schenk do not consider multiset semantics of SPARQL
in their translations.
   Perez et al. \cite{10145} gave the first formal
treatment of multiset semantics for SPARQL.
Angles and Gutierrez \cite{90006}, Polleres \cite{90821} and 
Schmidt et al. \cite{90380} extended the set semantics
to multiset semantics using this idea. 
 Kaminski et al \cite{Kaminski} considered
multisets in subqueries and aggregates in SPARQL. 
In none of these works was addressed the goal of characterizing the
multiset algebraic and/or logical structure of the operators in SPARQL.

\medskip

We studied the multiset semantics of the core SPARQL patterns, in
order to shed light on the algebraic and logic structure of them.
In this regard, the discovery that the core fragment of SPARQL
patterns matches precisely the multiset semantics of Datalog as defined by
Mumick et al. \cite{90820} and that this logical structure corresponds
to a simple multiset algebra, namely the Multiset Relational Algebra
(MRA), builds a nice parallel to that of classical set relational
algebra and relational calculus. Contrary to the rather chaotic
variety of multiset operators in SQL, it is interesting to observe that in
SPARQL there is a coherent body of multiset operators. We think
that this should be considered by designers in order to try to keep
this clean design in future extensions of SPARQL.

Last, but not least, this study shows the complexities and challenges
that the  introduction of multisets brings to query languages,
exemplified here in the case of SPARQL.

\paragraph{Acknowledgments.} The authors have funding from Millennium
Nucleus Center for Semantic Web Research under Grant NC120004. The
authors thank useful feedback from O. Hartig and anonymous reviewers.

\bibliographystyle{abbrv}
\bibliography{referencias}
\end{document}